\documentclass{article}
\usepackage{amsmath}
\usepackage{amsthm}
\usepackage{amssymb}
\usepackage{amsfonts}
\usepackage{xspace}
\usepackage{color}
\usepackage[]{todonotes}
\usepackage{url}
\usepackage{bm}
\usepackage{algorithm,algorithmic}

\newtheorem{theorem}{Theorem}
\newtheorem{lemma}[theorem]{Lemma}
\newtheorem{corollary}[theorem]{Corollary}

\newtheorem{proposition}[theorem]{Proposition}
\newtheorem{definition}[theorem]{Definition}

\newtheorem{example}[theorem]{Example}
\newtheorem{remark}[theorem]{Remark}

\newtheorem{acknowledgment}[]{Acknowledgment}

\newcommand{\mprimary}{\mathfrak{m}\text{-primary}}
\newcommand{\m}{\ensuremath{\mathfrak{m}}}
\newcommand{\ideal}[1]{\left\langle #1\right\rangle}
\newcommand{\NN}{\mathbb{N} }
\newcommand{\KK}{\mathbb{K} }
\newcommand{\ann}{\ensuremath{\mathrm{ann}}}
\newcommand{\anni}{\ann(\ell)}
\newcommand{\ssl}{\mathcal{S}_\ell}
\newcommand{\border}{\mathcal{B}_\ell}
\newcommand{\ps}{\KK[[\partial]]}
\newcommand{\qzeta}{Q_\zeta}
\newcommand{\pspol}{\KK[\partial]}
\newcommand{\iperp}{I^{\perp}}
\newcommand{\mM}{\ensuremath{\mathsf{M}}}
\newcommand{\quot}[2]{{\raisebox{.1em}{$#1$}\left/\raisebox{-.1em}{$#2$}\right.}}
\title{On Recurrence Relations of Multi-dimensional Sequences}

\author{Hamid Rahkooy\\
  Mathematical Institute, University of Oxford\\
\tt{rahkooy@maths.ox.ac.uk}}

\begin{document}

\maketitle

\begin{abstract}
  In this paper, we present a new algorithm for computing the 
  linear recurrence relations of multi-dimensional sequences. 
  Existing algorithms for computing these relations
  arise in computational algebra and include constructing structured 
  matrices and computing their kernels. The challenging problem is
  to reduce the size of the corresponding matrices. 
  In this paper, we show how to convert the problem of computing 
  recurrence relations of multi-dimensional sequences into 
  computing the orthogonal of certain ideals as subvector spaces of 
  the dual module of polynomials. We propose an algorithm using 
  efficient dual module computation algorithms. We present a complexity 
  bound for this algorithm, carry on experiments using Maple implementation, 
  and discuss the cases when using this algorithm is much faster 
  than the existing approaches.
\end{abstract}

\section{Introduction}

  Sequences are well-known mathematical objects that appear in many
  areas of science. Recurrence relations are a standard presentation
  of sequences. A recurrence relation is a relation between a finite
  number of the elements of the sequence with fixed coefficients such
  that the relation is preserved under any shift of the indices of the
  elements of the sequence.  The problem of finding recurrence
  relations of a given sequence is a classic problem.

  Considering sequences over a field as functions from natural numbers
  to a field, a natural generalization into \textit{multi-dimensional
    sequences} is the set of functions from $\NN^n$ into a field, for
  a positive integer $n$.  Recurrence relations of multi-dimensional
  sequences should be defined in such a way that they are preserved
  under \textit{shifts in any direction}. As a monoid, $\NN^n$ is
  isomorphic to the monomials in $n$ variables, hence every element of
  a multi-dimensional sequence is assigned to a monomial, which gives
  a hint to define recurrence relations of multi-dimensional
  sequences using multivariate polynomials, and then shifts in
  different directions would correspond to multiplication by
  monomials. A polynomial is defined to be a recurrence relation for a
  multi-dimensional sequence if the sum of the multiplications of its
  coefficients by the elements of the sequence corresponding to the
  relevant monomials is zero, and moreover, every multiplication of
  the polynomial with a monomial must have the same property.
  Recurrence relations of a multi-dimensional sequence form an
  ideal. We study the problem of finding
  recurrence relations \textit{efficiently}.
  
  In commutative algebra, multi-dimensional sequences correspond to
  the elements of the dual module of polynomials. Given a
  multi-dimensional sequence, the corresponding element of the dual
  module is defined by its value on a monomial being the value of 
  the sequence at that monomials, and then linearly extending the
  definition to polynomials. A recurrence relation of a
  multi-dimensional sequence is then an annihilator of the
  corresponding element of the dual module.

  Dual module and its computational aspects have been studied by
  Macaulay in his seminal work~\cite{mac}, where he introduced
  \textit{inverse systems} and presented an algorithm that takes the
  generators of an ideal and computes a set of generators of the
  corresponding subspace in the dual module.  Macaulay's algorithm
  considers \textit{enough many} multiplications of the given
  generators of the ideal and finds linear relations between them. The
  matrices associated with such linear systems are quasi-Toeplitz.  In
  algebraic geometry, the well-known Gorenstein schemes/algebras
  are tightly related to dual modules~\cite{Eisenbud} and the recurrence
  relations of multi-dimensional sequences can be used for
  computations on Gorenstein
  algebras~\cite{huneke-gorenstein,eisenbud-buchsbaum,elias-rossi}.

  A first approach for computing recurrence relations of a
  multi-dimensional sequence is to consider the vector of the first $s$  elements of the sequence and the vectors of its
  shifts, and then find the linear relations between those
  vectors. This leads to solving a linear system of equations whose 
  corresponding matrix is a Hankel matrix.  Assuming that the first $s$ elements of a sequence is enough for determining the ideal of recurrence relations,  then the corresponding Hankel matrix will be of size $s$, and hence, a first complexity bound for finding the
  kernel of this matrix is $\mathcal{O} (s^\omega)$, where $\omega$ is
  the constant for matrix multiplication. The structure of the matrix can help with obtaining a better complexity bound.
  Multi-dimensional sequences and computing their linear recurrence
  relations using the latter approach appeared in Sakata's work~\cite{sakata90,Sakata88:0,sakata89,sakata2009,sakata91} as 
  a generalisation of Berlekamp and Massey's work on shift register and
  coding~\cite{massey,berlekamp}. 

  In the 1990s and later, Marinari, Mora, M\"oller and Alonso studied 
  the problem of finding a Gr\"obner basis of the ideal for which the dual
  vector space is given~\cite{alonso-marinari-mora-duality,mari-mora-moel-functionals,mari-mora-moel-mult,mari-grazia-mora-moel-duality}. They
  showed that the problem of finding the annihilator of a given 
  subspace of the dual module is a generalisation of several
  well-known problems in computational algebra such as the FGLM
  Problem~\cite{fglm} and Buchberger-M\"oller
  algorithm~\cite{bb-moeller}. They presented two FGLM-style
  algorithms that traverse through the monomials, looking for their
  linear combinations that are killed by the given
  dual module elements. They also show an application on Hilbert
  scheme computation. The algorithms presented in~\cite{mari-mora-moel-functionals} have complexity bound $\mathcal{O} (nt^3+nft^2)$, where $t$ is the number of given linear maps/sequences and $f$ is the cost of evaluation maps on monomials. The algorithms assume that the input linear maps span
  the orthogonal vector space of the annihilator as a subspace of the dual module. A collection of the material on the duality problem can be found in Part VI of Mora's Book~\cite{mora-book2}.

  Computing recurrence relations of multi-dimensional sequences (or equivalently computing the annihilator of the dual module elements) is an active research problem and there are several recent publications on this problem.  
  Berthomieu, et al.   presented an FGLM-style algorithms 
  with the complexity bound
  $\mathcal{O}(r(r+|G|)m)$~\cite{berthomieu-boyer-faugere-jsc17}, where $G$ is the reduced Gr\"obner basis of the annihilator and m is the number of monomials that are less than the leading monomials of $G$ (and hence $m$ is the size of the border basis), and an improvement of this algorithm with the complexity bound $\mathcal{O}((r+|G|)^l)$~\cite{berthomieu-division-jsc22,berthomieu-guessing-jsc22}, where $l$ is the size of the Minkowski sum of the staircase with itself.

  Neiger, Rahkooy \& Schost presented an algorithm based on projection,
  with complexity $\mathcal{O}(ntBr^3)$, assuming that the $t$ input sequences satisfy the same condition as Mora et al.'s algorithm, and where $B$ is a bound on the degree of the minimal polynomial of the variables and $n$ is the number of the variables~\cite{bms-neiger-rahkooy-schost-casc}.  
  Computing the annihilator is a special case of Mourrain's algorithm
  for computing the border basis of Gorenstein Artinian algebras in
  \cite{mourrain-gorenstein}, where he applied his algorithm in 
  Tensor Decomposition Problem \cite{mourrain-tensor-decomp}. The 
  complexity of this algorithm is bounded by $\mathcal{O}(nsr^2)$
  \cite{mourrain-gorenstein}.

  In this paper, we propose a new method for computing recurrence relations of multi-dimensional sequences, or equivalently, finding the annihilator of a finite set of the dual module of polynomials. 
   We convert the problem of finding the annihilator into the problem of computing the orthogonal space of a certain ideal in the dual module, as a vector space. This conversion is inspired by the relationship between the
  (quasi-)Hankel matrices for computing the annihilator of given dual module elements and the (quasi-)Toplitz matrices used for computing the orthogonal space of the annihilator ideal as a subspace of the vector space of the dual elements. This enables us to exploit the existing algorithms for computing the orthogonal of an ideal as a subspace of dual space of polynomials. 

  Apart from the classical algorithm of Macaulay~\cite{mac}, Mourrain~\cite{mourrain1997} and Zeng~\cite{zeng} have introduced faster algorithms
  for computing the orthogonal space of an ideal. Our method can be applied to any of these algorithms and convert them into an algorithm for computing the recurrence relations of given sequences. We have considered Mourrain's algorithm~\cite{mourrain1997},  
  and optimisations on it~\cite{mantzaflaris-rahkooy-zafeirakopoulos} and introduced an algorithm for computing recurrence relations. 
  We have implemented our algorithm in \textsc{Maple}, extending the package in~\cite{angelos,mantzaflaris-rahkooy-zafeirakopoulos}. We prove the
  complexity bound $\mathcal{O}((s-r)^3+ns)$ for our algorithm, where $s$ is the number of given elements of the sequence, $n$ is the dimension of the ambient space and $r$ is the dimension of the quotient space of the polynomial ring with respect to the annihilator ideal. The complexity bound for our algorithm compared to the complexity bounds of the algorithms in~\cite{mourrain-gorenstein} shows that our algorithm is faster when $s-r$ is small, while it is slower when $s-r$ is large and the complexity bounds agree when $s=r$. 
  
  The existing algorithms for computing recurrence relations/annihilators can be considered FGLM-style as they traverse through the monomials starting from the smallest. Our algorithm is novel in the sense that it
  starts from the largest monomial for which the associated value of the sequence is given and increments towards the smallest monomial. 
  In addition, similar to the FGLM, the other algorithms construct matrices
  whose columns are labelled with the monomials, however, our
  algorithm labels the columns of the matrices constructed by the sum of several monomials, which reduces the number of columns. This property is
  inherited from Mourrain and Zeng's algorithms for computing the dual of a given ideal~\cite{mourrain1997,zeng}.

  The structure of this paper is as follows. In Section~\ref{section-prelim}, we introduce the notation and briefly introduce the required preliminaries on multi-dimensional sequences and their recurrence relations. We review the connection between multi-dimensional sequences and the dual module of polynomials as well as polynomial differential operators. Then we discuss the existing algorithms in the literature for computing recurrence relations of multi-dimensional ideals as well as the orthogonal of an ideal. In
  Section~\ref{section-algorithm}, we present the main definition of this paper and
  the main proposition for converting annihilator computation into the computation of the orthogonal of an ideal. We also explain
  the idea behind our algorithm via examples and present our algorithm.
  Section~\ref{section-complexity} is devoted to the complexity
  analysis of our algorithm, discussion of its advantages and
  disadvantages, comparison with other existing algorithms and also experiments via our Maple implementation.

\section{Preliminaries}\label{section-prelim}

In this section, we briefly introduce the required preliminaries.
Throughout this paper, we use the following conventions and notations. $\KK$
denotes an algebraically closed field and  
$\KK[X]=\KK[x_1,\ldots,x_n]$ denotes the ring of polynomials in variables $x_1,\ldots,x_n$, where $n$ a fixed integer, $I \trianglelefteq \KK[X]$ denotes that $I$ is an ideal in $\KK[X]$. For $\alpha = (\alpha_1,\ldots,\alpha_n)\in \NN^n$, $X^\alpha$ denotes the
monomial $x_1^{\alpha_1}x_2^{\alpha_2}\cdots x_n^{\alpha_n}$, and the
coefficient $c_\alpha$ denotes $c_{\alpha_1,\ldots,\alpha_n}$. So a
polynomial $f(x_1,\ldots,x_n) \in \KK[X]$ will be written in the form
of $f(X) = \sum\limits_{\alpha \in A} c_\alpha X^\alpha$, where $A$ is
a finite subset of $\NN^n$.
  
\subsection{Annihilator of $n-$dimensional Sequences}
The definitions and properties in this subsection can be 
found in more details  in \cite{bms-neiger-rahkooy-schost-casc}.

 We generalize the definition of a sequence into an $n-$dimensional sequence.
 \begin{definition}[$n-$dimensional Sequence \cite{bms-neiger-rahkooy-schost-casc}]
 An $n-$dimensional sequence over a field $\KK$ is an ordered set
 $\ell = (\ell_\alpha)_{\alpha \in  \NN^n}$, where 
 $\ell_\alpha \in \KK$, for all $\alpha$. We denote the set of all $n-$dimensional sequences 
 over $\KK$ by $\KK^{\NN^n}$. 
 \end{definition}
 
 Throughout this paper, $n$ is fixed and for convenience, we call an 
 $n-$dimensional sequence simply a sequence.

 $\KK^{\NN^n}$ is a $\KK[X]-$module. To see this, take an $n-$dimensional 
 sequence $\ell = (\ell_\alpha)_{\alpha \in  \NN^n}$ and a monomial 
 $X^\beta$ in $\KK[X]$. Define 
 $X^\beta \cdot \ell = (\ell_{\alpha+\beta})_{\alpha \in \NN^n}$ 
 ($\ell_\gamma=0 \ \forall \gamma < \beta$), and extend the definition linearly 
 to the polynomials as follows. Let 
 $f(X) = \sum\limits_{\gamma \in B \subseteq \NN^n} 
  f_\gamma X^\gamma \in \KK[X]$,  where $B$ is a finite subset of $\NN^n$, 
  and   define 
      $f(X) \cdot \ell = 
      ( \sum\limits_{\gamma \in B \subseteq \NN^n} 
      f_\gamma \ell_{\gamma+\alpha} )_{\alpha \in \NN^n}$.
 One can check that this 
 multiplication turns $\KK^{\NN^n}$ into a $\KK[X]-$module.
 \cite{bms-neiger-rahkooy-schost-casc}. 
 
 \begin{definition}[Annihilator of a Sequence 
 \cite{bms-neiger-rahkooy-schost-casc}, Section 2.1]
 Let $\ell=(\ell_\alpha)_{\alpha \in \NN^n}$ be an $n-$dimensional sequence.
 The  annihilator of $\ell$ is defined to be
 $\anni = \{ f(X) \in \KK[X] | f(X) \cdot \ell=0 \}$.  
 For several $n-$dimensional sequences $\ell_1,\dots,\ell_s$, the 
 annihilator is defined to be  
 $\ann(\ell_1,\dots,\ell_s) = \ann(\ell_1) \cap \cdots \cap \ann(\ell_s)$.
 \end{definition}
 One can easily check that the annihilator is an ideal in $\KK[X]$
 \cite{bms-neiger-rahkooy-schost-casc}.

 The set of $\KK-$linear functions from $\KK[X]$ to $\KK$ is denoted 
 by $\KK[X]^*$. For all $\ell^* \in \KK[X]^*$ and monomial $X^\alpha$, 
 define $X^\alpha\cdot \ell^* :\KK[X] \rightarrow \KK$, by
 $X^\alpha \cdot \ell^* (X^\beta) = \ell^*(X^{\alpha+\beta})$ for a monomial 
 $X^\beta$, and extend the definition to polynomials. 
 Linearly extending the multiplication of a monomial by a 
 linear function into the multiplication of a polynomial by 
 a linear function, makes 
 $\KK[X]^*$ a $\KK[X]-$module. This module is called the dual module 
 of $\KK[X]$. 
 One can also consider $\KK[X]$ as a vector space over $\KK$. 
 Then $\KK[X]^*$ is the dual vector space of $\KK[X]$. 
 In this paper, we use the word 
 ``dual'' both for $\KK[X]^*$ as a module and a vector space, unless we 
 clearly specify one of the two.

 There exists an isomorphism between $\KK^{\NN^n}$ and $\KK[X]^*$ 
 as $\KK[X]$modules.     This isomorphism maps 
 every $\ell=(\ell_\alpha)_{\alpha \in \NN^n} \in \KK^{\NN^n}$
 into a linear function $\ell^* \in \KK[X]^*$, 
 such that $\ell^*$ maps every monomial $X^\alpha$ into 
 $\ell^*(X^\alpha):=\ell_\alpha$, and extends it linearly to 
 every polynomial 
 $f(X) = \sum\limits_{\alpha \in A\subset \NN^n} c_\alpha X^\alpha$, i.e.,
 $\ell^*(f(X))= \sum\limits_{\alpha \in A} c_\alpha \ell^*(X^\alpha) 
 = \sum\limits_{\alpha \in A} c_\alpha \ell_\alpha$, where $A$ is a finite 
 subset of $\NN^n$.
 Conversely, to a linear function $\ell^* \in \KK[X]^*$, 
 one can associate a sequence 
 $\ell=(\ell_\alpha)_{\alpha \in \NN^n} \in \KK^{\NN^n}$, such that 
 $\ell_\alpha:=\ell^*(X^\alpha)$
 \cite{bms-neiger-rahkooy-schost-casc}.
 The above isomorphism can be considered as a vector space isomorphism 
 as well.
     
\subsection{Dual Module of Polynomials}\label{prelim-dual}

Required background on computations on dual module of polynomials 
 is given below from \cite{mourrain1997,mourrain-tensor-decomp}. 

 Assume that $\alpha = (\alpha_1,\ldots,\alpha_n) \in \KK^n$, and let
 $\partial^\alpha \in \KK[X]^*$ be the differential operator whose action on 
 a monomial 
 $X^\beta :=\prod\limits_{i=1}^n x_1^{\beta_1}\ldots x_n^{\beta_n}$ 
 is defined to be $\partial^\alpha ( X^\beta )= 1$ if $\alpha = \beta$, 
 and zero otherwise.  Extending the above action linearly to the polynomials, 
 define $\ps$ to be the set of power series of the differential operators 
 acting on the polynomials in $\KK[X]$. 
 Let $\partial^\alpha$ be a monomial differential operator in $\ps$ and 
 $X^\beta$ and $X^\gamma$ be monomials in $\KK[X]$. 
 Define $X^\beta \cdot \partial^\alpha := \partial^{\alpha-\beta}$.
 Again, extending this definition linearly to the multiplication of 
 a polynomial by a  differential power series, one can see that 
 $\ps$ is a $\KK[X]-$module. More details on the above can be found in 
 \cite{mourrain1997}.
 
 It is well-known that the $\KK[X]-$module $\ps$ is isomorphic to 
 $\KK[X]^*$, the dual of $\KK[X]$. To see this isomorphism,
 first of all, note that by definition, a power series of differential 
 operators is in $\KK[X]^*$. On the other hand, let $\ell^* \in \KK[X]^*$. For every monomial 
 $X^\alpha \in \KK[X]$, let $\ell_\alpha:=\ell^*(X^\alpha) $. 
 Define 
 $p_\ell(\partial):=\sum\limits_{\alpha \in \NN^n} 
 \ell_\alpha \partial^\alpha \in \ps$. Then $p_\ell(\partial)$ is the image 
 of $\ell^*$ under this isomorphism. More details 
 can be found in, e.g.,\cite{mourrain1997,mari-mora-moel-mult}. 
 
 Consider an ideal $I$ as a submodule (or sub-vector space) of the  module 
 (vector space) of polynomials. Then the orthogonal of $I$ in the dual 
 module (vector space) is defined to be 
 $I^\perp := \{ \Lambda \in \KK[X]^* | 
 \Lambda(f) =0 \ \ \ \forall f \in I\}$ \cite{mourrain1997}. 
 One can easily  check that $\iperp$ is a sub-module (subsapce) of 
 the dual module (vector space). 
 
 Having the isomorphism $\ps \simeq \KK[X]^*$, by the members of 
 $\iperp$ (or $\iperp$ itself), we refer to the linear maps in the dual 
 module as well as the corresponding differential operators. By computing 
 $\iperp$ we mean computing a basis for the corresponding subsapce of $\ps$, 
 consisting of differential power series (polynomials) operators. 

 For a positive integer $t$, let $\iperp_t$ denote the set of differential 
 polynomials in $\iperp$ that have degree at most $t$. $\iperp_t$ is a 
 subspace of $\iperp$, which we call the orthogonal of $I$ up to degree $t$.
 For example $\iperp_0$ only includes $1_0$, the evaluation functional at 
 $0$. Obviously $\iperp_t \subseteq \iperp_{t+1}$. Mora et al. have
 proved the following for zero dimensional ideals.
 \begin{theorem}[Mora, et al. 1993, \cite{mari-mora-moel-mult}, Theorem 3.1]~\label{thm:mora}
   There is a one to one correspondence between finite dimensional
   subspaces of $\KK[\partial]$ that are closed under derivation and
   $\mprimary$ ideals in $\KK[x_1,\dots,x_n]$, where
   $\m = \ideal{x_1,\ldots,x_n}$.  Moreover, If $I$ is an $\mprimary$
   ideal, then $\iperp$ is a subspace of $\pspol$, which includes
   polynomials instead of power series.
\end{theorem}

In other words, the orthogonal of an $\mprimary$ ideal hass a finite
basis whose elements are polynomial differential
operators. Conversely, the annihilator of a finite set of sequences
with finite support is an $\mprimary$ ideal.  Also there exists an
upper bound $\mathcal{N}$ on the degree of the elements of $I^\perp$,
called the \textit{nil-index} of $I$.  For more details on nil-index,
refer to \cite{mourrain1997}.  Therefore, the orthogonal of an
$\mprimary$ ideal is a finite dimensional sub-vector space of the dual
space and $\iperp_\mathcal{N} = \iperp_{\mathcal{N}+1} = \ldots$.
Following the above, the algorithms proposed for computing $\iperp$,
by Macaulay \cite{mac}, Mourrain \cite{mourrain1997} and Zeng
\cite{zeng}, start from $\iperp_0$ (which is $\ideal{1_0}$ for
$\mprimary$ ideals), and compute
$\iperp_1 \subseteq \iperp_2 \subseteq \ldots$ incrementally.  For
$\mprimary$ ideals, the existence of nil-index guarantees that this
chain is finite.
 
 Let $\KK[X]_t$ denote the vector space of polynomials in $\KK[X]$ 
 of degree at most $t$. 
 For computing $\iperp_t \subseteq \KK[X]_t$,  Macaulay's algorithm considers 
 a polynomial differential operator with indeterminate coefficients 
 (say $c_\alpha$, for $\alpha \in \NN^n$), 
 whose support includes all monomials of  degree at most $t$, i.e., 
 $p(\partial) = \sum\limits_{|\alpha|\leq t} 
 c_\alpha \partial^\alpha \in \pspol_t$. 
 One can easily check that $p(\partial) \in \iperp$ iff 
 $p(\partial)\cdot X^\beta f_i(X) =0$ for all monomials $X^\beta$.
 This leads to a Toeplitz system of linear equations in 
 indeterminates $c_\alpha$
 of size $\binom{t+n-1}{t}$ (the number of monomials of degree 
 at most $t$). 
 The complexity of solving this system 
 is $O\left( \binom{t+n-1}{t}^\omega\right)$, where $\omega$ is the constant 
 in the complexity of matrix multiplication.
 
 In order to speed up Macaulay's algorithm, Zeng and Mourrain took 
 advantage of the so called \textit{closedness property} of the 
 orthogonal of an ideal,
 which means that 
 the orthogonal of an ideal is closed under differentiation, 
 i.e., $\partial_i \cdot p(\partial) \in \iperp_{t-1}$, 
 for all $p(\partial) \in \iperp_t$, where 
 $\partial^\alpha \cdot \partial^\beta$ is defined to be 
 $\partial^{\beta-\alpha}$, for all $\alpha, \beta \in \NN^n$, and the 
 definition is linearly extended to 
 $\partial^\alpha \cdot p(\partial)$.
 Using this property, Mourrain's algorithm (which we refer to as the 
 integration method) \textit{integrates} members of a basis of 
 $\iperp_{t-1}$ in order to compute a basis for $\iperp_t$.
 
 \begin{theorem}[Mourrain, 1997, \cite{mourrain1997}]\label{them-integration}
  Assume \ that \ $I$ \ is \ an \ ideal \ generated \ by $f_1(X),\ldots,f_m(X)$ and  
  $t > 1$ is an integer. Let $p_1(\partial),\ldots,p_s(\partial)$ be a basis of $\iperp_t$.
  Then the elements of $\iperp_{t+1}$ with no constant terms 
  (in $\partial_i$) are of the form
  \begin{equation}\label{lambda}
  p(\partial) = \sum\limits_{i=1}^{s} \sum\limits_{k=1}^{n} \lambda_{ik} 
  \int_k p_i(\partial_1,\ldots,\partial_k,0,\ldots,0),
  \end{equation}
  such that 
  \begin{enumerate}
   \item $\forall \ \ 1 \leq k \leq m, \ p(\partial)( f_k(X) )= 0$  \label{cond1}
   \item $\forall \ \ 1 \leq k < l \leq n,  \   \sum\limits_{i=1}^{m} 
   \lambda_{ik}\partial_l\cdot p_i(\partial)  - \sum\limits_{i=1}^{m} 
   \lambda_{il}\partial_k \cdot p_i(\partial) =0$. \label{cond2}
  \end{enumerate}
 \end{theorem}
 For the details of the above theorem and the integration method, 
 refer to \cite{mourrain1997}. 
 
 As computing the orthogonal of an ideal is mostly used in investigating the 
 multiplicity structure of an isolated point, the above theorem and the 
 resulting algorithm is mostly used for ideals with an $\mprimary$ component. 
 However, the above theorem does not require $I$ to be $\mprimary$, 
 neither to have an $\mprimary$ component, and one can use it in order to 
 compute a basis for 
 $\iperp_t$ for every $t$, incrementally. Note that when $I$ is not 
 $\mprimary$, then not necessarily $\iperp_0 = \ideal{1_0}$, but actually it 
 can be that $\iperp_0=\emptyset$. In any case, this should be checked at 
 the beginning of the algorithm. Later in Section \ref{section-algorithm}, 
 in order to use the integration method for computing the annihilator, 
 we will treat this case in Lemma \ref{degenerate}. 
 
 The following proposition gives us the complexity of the integration method 
 when the input is a generating set for an ideal with an isolated component, 
 and the isolated point itself.
 
 \begin{proposition}[Proposition 4.1, \cite{mourrain1997}]\label{integ-complexity}
 Let $I = \ideal{p_1,\ldots,p_m} \trianglelefteq \KK[x_1,\cdots,x_n]$, 
 with a $m_\zeta-$primary component $\qzeta$ such that the multiplicity of 
 $\zeta$ is $\mu$. Also, assume that 
 $\qzeta^\perp = \ideal{\beta_1(\partial),\cdots,\beta_{\mu}(\partial)}$.
 The total number of arithmetic operations for computing $\qzeta^\perp$    
 during the integration algorithm is bounded by $O((n^2 +m)\mu^3 +nm\mu C)$,
 where $C$ is the maximal cost for  computing 
 $\int_k \beta_j|_{\partial_{k+1}=\cdots=\partial_n=0} (p_i)$. 
 $C$ can be bounded by $O(n\mu L)$, where $L$ is the number of monomials 
 obtained by derivation of the monomials of $(p_i)$.
 \end{proposition}
 
 The multiplicity of the isolated point in the above proposition, i.e.,  
 $\mu$, is the size of the largest matrix that is 
 constructed at the final step of the algorithm, step $\mathcal{N}$, 
 where $\mathcal{N}$ 
 is the nil-index of the $\mprimary$ component. Obviously, this can 
 be much smaller than the size of the largest matrix in Macaulay's algorithm, 
 which is $\mathcal{N}+n-1 \choose \mathcal{N}$. 

  Macaulay and Mourrain's algorithms have been implemented in the 
 Maple package mroot by Mantzaflaris 
 \footnote{\url{https://github.com/filiatra/polyonimo}}, 
 and improvements have been made in 
 \cite{angelos,mantzaflaris-rahkooy-zafeirakopoulos}.
 Experiments in \cite{angelos,mantzaflaris-rahkooy-zafeirakopoulos} 
 show that most of the computation time is spent on solving the 
 linear systems obtained during the algorithms. 
 It has been confirmed in practice that the integration method 
 constructs  much smaller matrices than Macaulay's, hence, it is much faster.

 Algorithm 4.5.14 in \cite{kreuz-rob-comp-lin-alg}, computes the 
 \textit{Canonical Module} of a given local zero-dimensional affine 
 $\KK-$algebra. In our setting the canonical module is the dual module, 
 and this algorithm constructs the same linear system as in 
 Macaulay's algorithm. 

\subsection{Dual Module vs Annihilator}\label{dual-vs-ann}

Having introduced the required background on dual module computation in 
Subsection \ref{prelim-dual}, below we explain the relation between the 
dual module, power series of differential operators and sequences. 
We will use these relation for the rest of the paper and construct our 
algorithm based on them.

 The isomorphism $\ps \simeq \KK[X]^*$ that we saw in the previous 
 subsection, as well the isomorphism $\KK^{\NN^n} \simeq  \KK[X]^*$, 
 implies an isomorphism between sequences 
 and $\ps$. More precisely, this isomorphism maps 
 the $n-$dimensional sequence  
 $\ell=(\ell_\beta)_{\beta \in \NN^n} \in \KK^{\NN^n}$ into
 $p_\ell(\partial) = 
 \sum \limits_\beta \ell_\beta \partial^\beta \in \ps$, 
 and vice versa. 
 
 Having the above isomorphism, 
 if $p_\ell(\partial)$ is the differential power series corresponding to 
 the sequence $\ell$, then for every $f(X) \in \KK[X]$, we have that 
 $f(X) \cdot \ell =\left( p_\ell(\partial) \left(X^\alpha f(X)  \right) 
        \right)_{\alpha \in \NN^n}$, 
 where the left hand side is the sequence obtained via the module 
 multiplication of a polynomial and a sequence, while the right hand side 
 is the the sequence obtained by acting the corresponding differential 
 operator on multiplications of $f(X)$  by all monomials $X^\alpha$.
 More details on the above isomorphisms can be found in 
 \cite{mourrain-tensor-decomp,mourrain-gorenstein,mari-mora-moel-mult}. 
 
 To summarize, $ \KK^{\NN^n} \simeq \KK[X]^* \simeq \ps $.
 In this paper, for an $n-$dimensional sequence 
 $\ell=(\ell_\alpha)_{\alpha \in \NN^n}$, we denote the corresponding member 
 of the dual module of polynomials by $\ell^* \in \KK[X]^*$, and the 
 corresponding differential power series by $p_\ell(\partial) \in \ps$. 
 Therefore, for every monomial $X^\beta$, we have that 
 $\ell_\beta =  \ell^*(X^\beta) = p_\ell(\partial) (X^\beta)$. 
 From now on we use these three representations interchangeably. 
 
 \begin{definition}
   Let $\ell = (\ell_\alpha)_{\alpha \in \NN^n}$ be a sequence with an
   $\mprimary$ annihilator. We define the support of $\ell$ to be the
   set of non-zero monomials in $\anni^\perp$ and denote it by $\ssl$.
 \end{definition}

 \begin{remark}
   For all $X^\alpha \notin \ssl$, $\ell_\alpha =0$. In other words,
   the set of monomials appearing in $p_{\ell}(\partial)$ is a subset
   of $\ssl$.
 \end{remark}
 
 \begin{remark}
   Let $\ell = (\ell_\alpha)_{\alpha \in \NN^n} \in \KK^{\NN^n}$ be an
   $n-$dimensional sequence, and $p_\ell(\partial) \in \ps$ be its
   corresponding differential power series. Then by the definitions of
   annihilator and the orthogonal of an ideal, we have that
   $p_\ell(\partial) \in \anni^\perp$.  Therefore, if $\anni$ is an
   $\mprimary$ ideal, then by Theorem~\ref{thm:mora},
   $p_\ell(\partial)$ is a polynomial differential operator, hence it
   has a finite support and $\ell_\alpha =0$ for all
   $X^\alpha \notin \ssl$. 
 \end{remark}

 \begin{remark}
   $\anni^\perp \simeq \quot{\KK[X]}{\anni}$. In particular, if $\anni$ is $\mprimary$
   then for all $\ell \in \anni^\perp$, $\ssl$ is finite and
   moreover,
   \begin{equation}
     \ssl \subseteq \{ x^\alpha \in \quot{\KK[X]}{\anni} \mid x^\alpha \ne  0 \}.
   \end{equation}
 \end{remark}

 \begin{remark}
   If $\ell_1,\dots,\ell_s$ is a basis for $\anni^\perp$ as a vector
   space, then 
   \begin{equation}
     \cup_{i=1}^s \ssl{_i} \subseteq \{ x^\alpha \in \quot{\KK[X]}{\ann(\ell_i)} \mid x^\alpha
     \ne 0 \}.
   \end{equation}
\end{remark}

 In this paper,  the annihilator of the sequences are $\mprimary$, 
 and $\ssl$ is the set defined above.
 A sequence in this paper is given by its values over $\ssl$.
 If it is clear from the context, we may use $\ell(X^\alpha)$ instead of 
 $\ell_\alpha$.
 
\begin{definition}
Let $\ell = (\ell_\alpha)_{\alpha \in \NN^n} \in \KK^{\NN^n}$, with an 
$\mprimary$ annihilator, whose support is $\ssl$. 
The border of $\ell$ is defined to be  
$\border = \{ X^\beta | \frac{X^\beta}{x_i} \in \ssl \ \ \text{and} \ \ 
 X^\beta \notin \ssl, \forall i \ \ 1\leq i \leq n \}$ . 
\end{definition}

Obviously, $\border \subseteq \anni$.

\begin{remark}
  We note that the annihilator of sequences is tightly related to the
  {\it Inverse Systems} in Macaulay's work in \cite{mac}. In
  particular, the one-to-one correspondence in Theorem 21.6
  in~\cite{Eisenbud} is similar to the one-to-one correspondence in
  Theorem~\ref{thm:mora}. We explain this briefly using the notation
  in~\cite{Eisenbud}. First of all, one can easily check that the
  $\KK[X]-$module $T = \KK[x_1^{-1},\dots,x_n^{-1}]$ is isomorphic to
  $\KK[\partial]$. For given sequences $\ell_1,\dots,\ell_s$ with
  corresponding polynomial differential operators
  $p_{\ell_1}(\partial),\dots, p_{\ell_s}(\partial)$, consider their
  images in $T$ and let $M$ be the module generated by their
  image. One can check that the annihilator of this module in $\KK[X]$
  is exactly the annihilator of the sequences $\ell_1,\dots,\ell_s$. on
  the other hand, the annihilator of a zero-dimensional ideal $I$ in
  $\KK[X]$ with respect to $T$ is $\iperp$, the orthogonal of $I$.
\end{remark}  
\begin{remark}
 the inverse system of a
  sequence $\ell$ is the orthogonal of $\anni$, and also $\anni$ is
  the polynomial ideal orthogonal to the inverse system of $\ell$. For
  details we refer to \cite{mourrain-tensor-decomp}.
\end{remark}

\section{The Algorithm}\label{section-algorithm}
  In this section, we explain our algorithm for computing the annihilator of 
  a sequence (or several sequences), using the algorithms for computing the 
  orthogonal of an ideal, in particular, the integration algorithm.
  
 Given a sequence $\ell = (\ell_\alpha)_{\alpha \in \NN^n}$ 
 with its values over $\ssl$ described in the previous subsection, a first 
 algorithm for computing the annihilator is as follows.  Let 
 $f(X) = \sum\limits_{\gamma \in B \subseteq \NN^n} f_\gamma X^\gamma 
 \in \ann(\ell)$,  
 with symbolic coefficients $f_\gamma$. Then by the definition of the 
 annihilator, we have that 
 \begin{equation}\label{ann-lin-sys}
   \sum\limits_{\gamma \in B \subseteq \NN^n} 
 f_\gamma \ell_{\gamma+\alpha} =0 \ \ \  \forall \alpha \in \NN^n. 
 \end{equation}

 As the sequence is given over a finite set $\ssl$ such that for all 
 $\alpha \notin \ssl$ we have that $\ell_\alpha =0$, then equations 
 \ref{ann-lin-sys} yield a finite linear system of equations with $f_\gamma$ 
 as indeterminates. 
 One can consider the associated matrix $H_\ell$ labeled as follows. 
 Having set a term order on the monomials, the columns of $H_\ell$
 are labeled by the monomials in $\ssl$ in an ascending order. 
 Rows of $H_\ell$ are labeled by $X^\alpha \cdot \ell$, for $\alpha \in \ssl$. 
 Then the entry $(\alpha,\beta)$ in $H_\ell$ is actually 
 $\left( X^\beta \cdot \ell\right)_{\alpha} = \ell_{\alpha+\beta}$.
 This is indeed a \textit{Hankel} system of size 
 $s:=|\ssl|$. Hence, the complexity of this system can be bounded by 
 $s^{\omega}$, 
 where $\omega$ is the constant in the complexity of matrix multiplication.
 For more than one sequence, the annihilator can be computed similarly. 
 The resulting linear system in this case will be a quasi-Hankel system.

\subsection{Reciprocal of a Sequence}
 We define the reciprocal of a multivariate polynomial, similar to that 
 of a univariate polynomial, which plays an essential role in converting the
 annihilator computation into dual computation.
 
 \begin{definition}\label{rec}
 Let 
 $p(X) = \sum\limits_{\gamma \in A \subseteq \NN^n} p_{\gamma} 
 X^{\gamma} \in \KK[X]$, where $A$ is a finite subset of $\NN^n$, and 
 assume that $d_i = \deg_{x_i}(p)$ are the partial degrees of $p(X)$, 
 $i=1,\ldots,n$. Let $d=(d_1,\ldots,d_n)$ and $X^d:=\prod\limits_{i=1}^n 
 x_i^{d_i}$.
 The reciprocal of $p(X)$ is defined to be 
 $rec(p) = \sum\limits_{\gamma} p_{\gamma} X^{d-\gamma} 
 = X^d p(\frac{1}{x_1},\cdots,\frac{1}{x_n})\in \KK[X] $.

 Let $\ell = (\ell_{\lambda})_{\lambda \in \ssl \subseteq \NN^n}$ be a 
 sequence (with $\ssl$ as in the Preliminaries section), and 
 $p_\ell(\partial) =
 \sum\limits_{\lambda \in \ssl} \ell_{\lambda} \partial^{\lambda} \in \pspol$ 
 be its corresponding polynomial differential operator.
 Assume that $d_i = \deg_{\partial_i}(p_\ell(\partial))$ are the partial 
 degrees of $p_\ell(\partial)$, $d=(d_1,\ldots,d_n)$ and 
 $\partial^d:=\prod\limits_{i=1}^n \partial_i^{d_i}$.
 The reciprocal of $p_\ell(\partial)$ is defined to be 
 $rec(p_\ell(\partial)) = \sum\limits_{\lambda \in \ssl} 
 \ell_{\lambda} \partial^{d - \lambda} \in \pspol$.
 Also the reciprocal of $\ssl$ is defined to be
 $rec(\ssl):=\{X^{\alpha-\gamma}  | \ X^\gamma \in \ssl \}$, 
 i.e., the set of reciprocals of the monomials in $\ssl$,
 where $X^\alpha$ is as above.
 \end{definition}

The reciprocal of a sequence $\ell$ is defined to be
the reciprocal of its associated differential polynomial operator,
that is, $rec(p_\ell(\partial))$.
Below we present some properties of reciprocal.

\begin{lemma}~\label{lem:lcm}
  Let $m_1,\dots,m_s$ be monomials in $\KK[X]$. If
  $m = lcm(m_1,\dots,m_s)$, then
  $lcm(\frac{m}{m_1},\dots,\frac{m}{m_s}) = m$.
\end{lemma}
\begin{proof}
  By assumption, $m_i \mid m$ for all $i$. Then $\frac{m}{m_i} \mid m$
  for all $i$. Assume that n is such that $\frac{m}{m_i} \mid n$ for
  all $i$. Then $m \mid nm_i$ for all $i$. Therefore
  \begin{equation}
    m \mid gcd(nm_1,\dots,nm_s) = n gcd(m_1,\dots,m_s) = n.
  \end{equation}
  So $m$ is a common multiple of $\frac{m}{m_i}$ and therefore, $m$
  divides any other common multiple. Hence,
  $lcm(\frac{m}{m_1},\dots,\frac{m}{m_i}) = m$.
\end{proof}

\begin{proposition}\label{prop:reciprocal-reciprocal}
  Let $f = c_1m_1 + \dots + c_sm_s$ be a polynomial in $\KK[X]$ with
  $c_i$ as its coefficients and $m_i$ as its monomials, and
  $gcd(m_1,\dots,m_s)=1$. Then
  \begin{equation}
    rec(rec(f)) = f.
  \end{equation}
\end{proposition}
\begin{proof}
  From the assumption $gcd(m_1,\dots,m_s)=1$ and Lemma~\ref{lem:lcm},
  we have that $lcm(\frac{m}{m_1},\dots,\frac{m}{m_s}) =
  m$. Therefore,
  \begin{align}
    rec(rec(f)) & = m(c_1\frac{1}{m/m_1} + \dots + c_s\frac{1}{m/m_s})
    \\
    & = c_1m_1 + \dots + c_sm_s.
  \end{align}
\end{proof}

\begin{corollary}
  Let $f \in \KK[X]$ and $cont(f)$ denote the content of $f$, which is the greatest common divisor of the monomials in $f$. Write $f=cont(f)f'$, where $f'\in \KK[X]$ and the greatest common divisor of the monomials of $f'$ is one. Then  
  $rec(f)=rec(f')$.
\end{corollary}
\begin{proof}
  The proof follows from Proposition~\ref{prop:reciprocal-reciprocal}.
\end{proof}

 In order to explain the idea behind of our algorithm, we fix the following 
 notation for the rest of this paper. $\ell$ and $X^d$ are as in 
 Definition \ref{rec}, and $m:= deg(X^d)=d_1+\cdots+d_n$. 
 $H_\ell$ denotes the 
 Hankel matrix for computing the annihilator of $\ell$.
 Consider $p_\ell(\partial)$, the differential operator associated with 
 $\ell$, and replace $\partial$ with $X$. Then $p_\ell(X)$ is the 
 generating function of $\ell$. Take the reciprocal of $p_\ell(X)$ and 
 consider the ideal generated by it, i.e., $\ideal{rec(p_\ell(X))}$.
 Let $M_\ell$ denote the Macaulay 
 matrix for computing  $\ideal{rec_\ell(p_\ell(X))}^\perp_m$, up to degree 
 $m$.
 
 One can easily check that the column labeled by $X^{\gamma}$ in 
 $H_\ell$ is the same as the column labeled by 
 $\partial^{d - {\gamma}}$ in $M_\ell$. So, $H_\ell$ is a 
 submatrix of $M_\ell$, up to a permutation of columns. The columns of the 
 Macaulay matrix include all monomials that divide $X^d$, which include 
 $rec(\ssl)$. Let $\widetilde{M}_\ell$ be the matrix obtained by removing 
 the columns labeled by the the monomials not in $\ssl$. We call this matrix 
 the modified Macaulay matrix for $\ell$.  Then $H_\ell$ and 
 $\widetilde{M}_\ell$ have the same columns up to a permutation, which 
 implies that $Ker(\widetilde{M}_\ell) \simeq Ker(H_\ell)$.  

 \begin{remark}\label{D}
 Note that $Ker(\widetilde{M}_\ell)$ does not give the whole vector space
 $\ideal{rec_\ell(p_\ell(X))}^\perp_m$, but a subspace of it. We call this 
 subspace $D_\ell$ and use this notation for the rest of this paper. 
 Note that for all $\Lambda(\partial) \in D_\ell$, monomials of $\Lambda(X)$ 
 are in $rec(\ssl)$. 
 \end{remark}
 Obviously the reciprocal of a basis for $D_\ell$ is a basis for the 
 annihilator of $\ell$ as a vector space. 
 
 \begin{example}
 Let $\ell$ be given by its non-zero values as follows: 
 $\ell(1) = 1, \ell(x)=2,\ell(y)=2,\ell(x^2)=4,\ell(xy)=4,\ell(y^2)=4$.
 Then $R(X):=rec(\ell(X))= 4x^2 + 4xy+ 4y^2 + 2x^2y + 2xy^2 + x^2y^2$.
 The corresponding Hankel and matrix for computing the annihilator 
   is 
   \[
   H_\ell = \bordermatrix{
     & 1 & x & y & x^2 & xy & y^2 \cr
   1 & 1 & 2 & 2 & 4 & 4 & 4 \cr
   x & 2 & 4 & 4 & 0 & 0 & 0 \cr
   y & 2 & 4 & 4 & 0 & 0 & 0 \cr
   x^2 & 4 & 0  & 0 & 0 & 0 & 0 \cr 
   xy & 4 & 0  & 0 & 0 & 0 & 0 \cr 
   y^2 & 4 & 0 & 0 & 0 & 0 & 0 \cr 
   },
   \] 
   
   and the modified Macaulay matrix  is
   
   \[
   \mM_\ell = \bordermatrix{
    & \partial_x^2 & \partial_x \partial_y & \partial_y^2 & 
    \partial_x^2 \partial_y  & \partial_x \partial_y^2 
    & \partial_x^2 \partial_y^2 \cr
   R(X) & 4 & 4 & 4 & 2 & 2 & 1 \cr
   xR(X)  & 0 & 0 & 0 & 4 & 4 & 2 \cr
   yR(X) & 0 & 0 & 0  & 4 & 4 & 2 \cr
   x^2R(X) & 0 & 0  & 0 & 0 & 0 & 4 \cr 
   xyR(X) & 0 & 0  & 0 & 0 & 0 & 4 \cr 
   y^2R(X) & 0 & 0 & 0 & 0 & 0 & 4 \cr
   }.
   \]
   \end{example}
   
 As we have already mentioned beforehand, 
 $H_\ell$ is a Hankel matrix, while $\widetilde{M}_\ell$ is a Toeplitz 
 matrix. The above proposition and description below it show how to convert 
 these structured matrices into each other. The relation between structured 
 matrices, and their relation to duality has been further studied by 
 Mourrain and Pan in   \cite{mourrain-pan-struc-mat}.

 \begin{proposition}\label{dual-ann-prop}
 Let $\ell$ and $p_\ell(\partial)$ be as in Definition \ref{rec} 
 and $f(X) = \sum\limits_{\gamma \in A \subseteq \NN^n} 
 f_{\gamma} X^{\gamma} \in \KK[X]$  be a polynomial.
 Assume that $p_\ell(X) \in \KK[X]$ is the polynomial obtained by replacing 
 $\partial_i$ with $x_i$ in $p_\ell(\partial)$,
 and $f(\partial) \in \pspol$ is the polynomial differential operator 
 obtained by replacing $x_i$ with $\partial_i$ in $f(X)$. Then 
 \[ f(X) \in \anni \ \ \Leftrightarrow \ \ 
 rec(f(\partial)) \in  \ideal{ rec(p_\ell(X))}^\perp.\]
 \end{proposition}
 \begin{proof} 
 Let     $X^\alpha = \prod x_i^{\deg_i(p_\ell(\partial))}$ 
 and $\partial^\beta = \prod \partial_i^{\deg_i(f(X))}$. 
 $rec(f(\partial))= 
 \sum\limits_{\gamma \in A} f_{\gamma} \partial^{\alpha - \gamma}$, 
 and  $rec(\ell(X)) 
    = \sum\limits_{\lambda \in \ssl} \ell_{\lambda} X^{\beta - \lambda}$.
    
 \begin{align*}
     f(X) \in \anni  & \Leftrightarrow f(X) \cdot  \ell= 0 \\
      & \Leftrightarrow \forall \tau \in \NN^n \ \ 
      \sum\limits_{\gamma \in A}   \ell_{\gamma+\tau}f_{\gamma} =0 \\
      & \Leftrightarrow \forall \tau \in \NN^n \ \ 
      \sum\limits_{\gamma \in A}   f_{\gamma} \ell_{\gamma+\tau}=0  \\
     &  \Leftrightarrow \forall \tau \in \NN^n \ \
      \sum\limits_{\gamma \in \alpha-A}  
      f_{\alpha - \gamma} \ell_{(\alpha -\gamma) + \tau} =0,
 \end{align*}     
 where $\alpha - A:=\{\alpha -m \ | \ m \in A \}$. Then the above holds iff
 \begin{align*}
   & \forall \tau \in \NN^n \ \ 
 (\sum\limits_{\gamma \in \alpha-A} f_{\alpha-\gamma}\partial^{\gamma}) \cdot
   (\sum\limits_{\lambda \in \ssl} \ell_{(\beta - \lambda)+\tau} X^{\lambda}) =0 \\
     & \Leftrightarrow \forall \tau \in \NN^n \ \ 
     (\sum\limits_{\gamma \in \alpha-A} f_{\gamma}\partial^{\alpha - \gamma}) \cdot
     (\sum\limits_{\lambda \in \ssl} \ell_{\lambda} X^{(\alpha -\lambda)+\tau}) =0 \\
       & \Leftrightarrow  \forall \tau \in \NN^n \ \ 
     rec(f(\partial)) \cdot X^\tau rec(\ell(X)) =0 \\
    & \Leftrightarrow rec(f(\partial)) \in \ideal{rec(\ell(X))}^\perp.
 \end{align*}
\end{proof}
 
  The idea of our algorithm for computing the elements of the   annihilator 
  from the dual computation algorithms should be clear from the above 
  explanations and the proposition. Here we sketch the algorithms further.
  Consider $p_\ell(X)$, the generating function of the sequence $\ell$. 
  We remind 
  the assumption from Subsection \ref{prelim-dual} that $\anni$ is 
  $\mprimary$, which implies that the support of $p_\ell(X)$ is a finite 
  set $\ssl$, which means that $p_\ell(X)$ is a polynomial.
  Compute the reciprocal of $p_\ell(X)$. Then compute a basis for the 
  orthogonal of the 
  ideal generated by the reciprocal up to degree $m=deg(X^d)$.
  This basis will  consist of polynomial differential operators. 
  Note that for all $f(X)\in \anni$, we have that $deg(f(X)) = t$ iff 
  $deg(rec(p_\ell(X)))=deg(X^d)-t=m-t$. This is another indication that 
  we only need to compute the orthogonal up to degree $m$. 
  Then substitute $\partial$ with $X$ in all polynomial differential 
  operators in the basis of the orthogonal and obtain polynomials in $\KK[X]$.
  Finally take the reciprocal of these polynomials. 
  The resulting polynomials form a basis for the annihilator of $\ell$ as a 
  vector space.
  
  If we use Macaulay's algorithm in order to compute the orthogonal of the 
  ideal generated by the reciprocal of the generating function, we will 
  construct matrices as large as the Hankel matrices used for computing the 
  annihilator, which does not give any improvement in terms of the 
  computation time. However, the integration method by Mourrain is much 
  more efficient. 
 
  Let $I=\ideal{rec(p_\ell(\partial))}$. In order to 
  use Theorem \ref{them-integration} and the integration method, we need 
  to compute the degree zero elements of the orthogonal, i.e., 
  $\iperp_0$. If $I$ is an $\mprimary$ ideal, then $\iperp_0 = \ideal{1_0}$. 
  If $I$ is not  $\mprimary$, then either $1_0 \in \iperp_0$ or 
  $\iperp_0 = \emptyset$. One can easily check from Proposition 
  \ref{dual-ann-prop} that 
  $1_0 \in \iperp_0$ if and only if $X^d \in \anni$. 
  The latter can be checked at the beginning of the algorithm using below 
  lemma, so that further computation is avoided in this case.
   
 \begin{lemma}\label{degenerate} 
  Let $\ell$ and $X^\alpha$ be as above. If $\ell_\alpha \ne 0$, then 
  $\anni = \ideal{\border}$.
 \end{lemma}
 \begin{proof}
  Let $f(x) = \sum\limits_{X^{\lambda} \in \ssl} c_{\lambda}X^{\lambda} \in 
  \anni \setminus \border$. Put a term-ordering on the monomials and assume 
  that $\beta_0$ is the smallest among the monomials in $f(X)$, i.e., 
  $\forall j \ X^{\beta_0} < X^{\beta_j}$. As 
  $\forall j \ X^{\lambda} | X^\alpha$, then obviously 
  $X^{\beta_0} | X^\alpha$.
  Therefore, $\frac{X^\alpha}{X^{\beta_0}}f(X) \in \anni$. Let 
  $g(X)=\frac{X^\alpha}{X^{\beta_0}}f(X)$, and write it as 
  $g(X) = \sum\limits_{\lambda \in \ssl}^s c_{\lambda}\frac{X^{\lambda+\alpha}}
  {X^{\beta_0}} = c_{\beta_0}X^\alpha + $ other terms. Then
  $g(X) \cdot \ell = c_{\beta_0} \ell_\alpha+\sum\limits_{\gamma \ne \alpha}
  \ell_{\gamma}$. In the left hand side, for all $\gamma \ne \alpha$, we have 
  that $X^\alpha < X^\gamma$, which means that $X^\gamma \notin \ssl$. 
  Therefore $g(X) \cdot \ell = c_{\beta_0}\ell_\alpha =0$. But by the 
  assumption 
  of the lemma $\ell_\alpha \ne 0$, which implies that $c_{\beta_0}=0$ which 
  is a contradiction with the assumption that $X^{\beta_0}$ is the smallest 
  term in $f(X)$.
 \end{proof}
  
  Also as we already know that $\border \subseteq \anni$, we compute 
  $\border$ and add it to the annihilator before running the integration 
  method.

\subsection{The Algorithm}  
 Summarizing, we have the following algorithm.
 \vspace{5mm}

 \begin{algorithm}[H]\label{algorithm}
  \caption{$\mathsf{AnnihilatorViaDuality(\ell)}$ }
  {\bf Input:} $(\ell_{\gamma})_{\gamma \in \ssl }$, 
  for a finite set $\ssl$ \\
  {\bf Output:} A basis for $\anni$ as a sub-vector space of $\KK[X]$, 
  \begin{enumerate}
   \item Compute $p_\ell(X) = \sum\limits_{\gamma \in \ssl} 
    \ell_{\gamma} X^{\gamma} $, the generating function of $\ell$.
   \item Let $d=deg(p_\ell(X))$, the multi-degree of $p_\ell(X)$
      and $m=deg(X^d)$
   \item\label{border-comp} Compute $\border$ 
   \item If $\ell_\alpha \ne 0$ then {\bf return} $\border$, else 
   \item\label{recip-comp} Compute a basis $\Lambda_1(\partial),\ldots,\Lambda_s(\partial)$ 
   for the subspace $D_\ell$ of $\ideal{rec(p_\ell(X))}^\perp$ 
   (Remark \ref{D})   via the integration   method
   \begin{enumerate}
    \item ({\bf Optimization}) delete columns with labels not in 
     $\ssl \cup \border$ in the matrices for computing  
     \ \ $\ideal{L(X)}_d^\perp $ 
   \end{enumerate}
   \item $\anni = \border \bigcup 
   \{ rec(\Lambda_1(X)),\ldots, rec(\Lambda_s(X)) \}$
   \item {\bf return} $\anni$
   \end{enumerate}
 \end{algorithm}

  \vspace{5mm} 

  Let $\anni_{\geq t}$ be the elements of the annihilator with monomials 
  of degree at least $t$. If $f(X) \in \anni_{\geq t}$, then 
  $x_if(X) \in \anni_{\geq t+1}$ for all $i=1,\ldots,n$. This corresponds 
  to the closedness property for dual modules. More precisely, 
  \begin{remark}\label{closedness-annihiator}
   In analogy with \cite{mourrain1997}, the closedness property of the dual 
   is equivalent to the fact that the annihilator $\anni$ is closed   under 
   multiplication, i.e., if 
   $   \anni_{\geq d-(t-1)} = \ideal{b_1(X),\ldots,b_{s_{t-1}}(X)}$, then 
   \[ 
   b(X) \in \anni_{\geq d-t} \Leftrightarrow x_1b(X),\ldots,x_nb(X) 
   \in \anni_{\geq d-(t-1)}.
   \]   
  \end{remark}
  
  One can re-write Theorem \ref{them-integration} using the above notation 
  and Proposition \ref{dual-ann-prop} in order to obtain a formula for 
  $b(X) \in \anni_{\geq d-t}$ using a basis of $\anni_{\geq d-(t-1)}$. Then 
  replacing integration with division in Theorem \ref{them-integration}, 
  one can directly take advantage of the fact that the
  annihilator is closed under multiplication. However, the computations 
  via this formula will be exactly the same as the computations in 
  Algorithm \ref{algorithm}. Berthomieu \& Faug\`ere's 
  in \cite{berthomieu-faugere-issac18} use this idea in their 
  Polynomial-Division-Based algorithm. 
  
 \begin{remark}\label{{several-seq}}
 The algorithm can be easily generalized to compute the annihilator of 
 more than one sequence, say $\ell_1,\ldots,\ell_t$. In this case, 
 $\ssl$ should be replaced  with $\mathcal{S}=\bigcup_{i=1}^t \ssl^{(i)}$, 
 where
 $\ssl^{(i)}$ is the support of $\ell_i$. 
 Also in step \ref{border-comp}, $\border$, the border of a sequence should 
 be replaced with $\mathcal{B}=\bigcup\limits_{i=1}^t \border^i$, where 
 $\border^i$ is the border of the sequence $\ell_i$.
 In step \ref{recip-comp} of the 
 algorithm, 
 one should compute the subspace of 
 $\ideal{rec(p_{\ell_1}(X)),\ldots,rec(p_{\ell_t}(X))}^\perp$
 whose monomials are in the reciprocal of $\mathcal{S}$. 
 \end{remark}

  \begin{remark} 
  Although we assume that in our algorithm the annihilator 
 is an $\mprimary$ ideal, however, 
 the algorithm is correct for sequences whose annihilator is not 
 necessarily   $\mprimary$, but have the flat extension property as in 
 Mourrain's algorithm. Discussion on the flat extension property is out of 
 the scope of this paper. For the details we refer the 
 reader to \cite{mourrain-gorenstein}.
 \end{remark}
  
\section{Complexity and Experiments}\label{section-complexity}

In this section, we compare our algorithm with other existing
algorithms, both in terms of their complexity and experiments.

\subsection{Complexity of the Algorithm}

 Let us remind that for our algorithm the input is a sequence 
 $\ell= (\ell_\alpha)_{\alpha \in \ssl \subseteq\NN^n}$, in which the values 
 $\ell_\alpha$ are given over a finite set $\ssl$. Also  $s:=|\ssl|$ and 
 $r:=\dim_\KK \quot{\KK[X]}{\anni}$ and $D_\ell$ is as in Remark \ref{D}.

  \begin{theorem}\label{annividual-complxity}
   Algorithm~\ref{algorithm} is correct and the total number of
   arithmetic operations in Algorithm~\ref{algorithm} using the
   integration method can be bounded by
 \begin{equation}\label{complexity-formula}
   O\left(  n^2(s-r)^3+ns     \right).
 \end{equation}
 \end{theorem}
 \begin{proof}
   The correctness of the algorithm comes from
   Proposition\ref{dual-ann-prop} and Lemma~\ref{degenerate}.  For the
   complexity bound, we only compute the number of arithmetic
   operations in steps \ref{border-comp} for computing $\border$ and
   \ref{recip-comp} for computing a basis for $D_\ell$.  The cost of
   computations in other steps can be ignored in comparison with those
   two steps.
  
 Computing step {border-comp} is basically reading off the members of $\border$, for 
 which the cost is $|\border| \leq n |\ssl| = ns$. This is the second summand 
 in formula \ref{complexity-formula}.
  
 In order to compute the complexity of computing a basis for $D_\ell$ using 
 the integration method, below we will explain a modification of 
 Proposition 4.1 in \cite{mourrain1997}  on the complexity of the integration 
 method for ideals with an isolated point into an arbitrary ideal. This 
 comes directly from the proof of Proposition \ref{complexity-formula} in 
 \cite{mourrain1997}.
  
 With the input an ideal and an isolated point $\zeta$ of the ideal, the 
 integration method computes a basis for $\qzeta^\perp$, where $\qzeta$ 
 is the primary component of the ideal corresponding to $\zeta$.
 The size of the basis for $\qzeta^\perp$ is 
 $\mu =\dim \quot{\KK[X}{\qzeta} = \dim \qzeta^\perp$, the multiplicity of 
 $\zeta$. This is where $\mu$ in Proposition 4.1 in \cite{mourrain1997} comes 
 from.
 AnnihilatorViaDuality instead computes a basis  for the vector space 
 $D_\ell$. 
 Let $\{ \Lambda_1(\partial),\ldots,\Lambda_t(\partial) \}$ be such a basis.
 Then $\mu$ replaced with $t$ in Proposition 4.1 in \cite{mourrain1997}
 gives the complexity of computing a basis for $D_\ell$ using the 
 integration  method. Below we prove that $t=s-r$. 
  
 Consider $\anni$ as a vector space over $\KK$. By Proposition 
 \ref{dual-ann-prop}, members of $D_\ell$ are in one-to-one correspondence 
 with the the members of 
 $\anni' = \{ f(X) \in \anni \ | \ Supp(f(X)) \subseteq \ssl  \}$.
 $\anni'$ is a subspace of $\anni$, and 
 $\{ f_1(X) = rec(\Lambda_1(X)), \ldots, f_t(X)=rec(\Lambda_t(X))\}$ is
 a basis of $\anni'$.
 One can easily check that 
 $\{ \Lambda_1(\partial),\ldots,\Lambda_t(\partial) \}$ is  a  linearly 
 independent set in $D_\ell$ if and only if $\{f_1(X),\ldots,f_t(X)\}$ is  
 a linearly independent set in $\anni$. $\anni$---as an ideal--- is 
 generated by 
 $\{ f_1(X),\ldots,f_t(X)\} \bigcup \border$. Let $W = \ideal{\ssl}$ be the 
 vector space spanned by the monomials in $\ssl$. Then $\anni'$ is a subspace 
 of $W$, and therefore we have 
 $\dim(\anni') +  \dim (\quot{W}{\anni'}) = \dim(W)$. 
 But $\dim(\quot{W}{\anni'}) = \dim(\quot{\KK[X]}{\anni'})$, 
 because all 
 monomials in $\KK[X] \setminus \ssl$ are zero in $\quot{\KK[X]}{\anni'}$.
 Therefore  $\dim(\anni') = s-r$.
 
 Replacing $\mu$ with $s-r$ in Proposition 4.1 in \cite{mourrain1997}, the 
 complexity of computing a basis for $D_\ell$ can be bounded by 
 $O(n^2 (s-r)^3 + n(s-r) C)$, where $C$---in analogy with Proposition 
 4.1 in \cite{mourrain1997}---can be bounded by $O(n(s-r)L)$, where $L$ is the 
 number of monomials obtained by derivation 
 of 
 the monomials of $rec(p_\ell(X))$ except for those that are out of a basis 
 for $\quot{\KK[X]}{\anni}$. This can be bound by $s-r$. 
 So complexity of computing $C$ can be bounded by $n^2(s-r)^3$, and 
 therefore the complexity of computing a basis for $D_\ell$ is at most
 $O(n^2 (s-r)^3)$. This proves the proposition.
 \end{proof}

 We remind that the complexity of computing a basis for $D_\ell$ using 
 Macaulay's algorithm is the same as computing the annihilator via Hankel 
 matrix, which can be bounded by $s^\omega$. 
 The complexity of Mourrain's algorithm as well as Berthomieu \& 
 Faugere's Polynomial-Division-Based algorithm is $nsr^2$, while the
 complexity of AnnihilatorViaDuality algorithm is $n^n(s-r)^3 +ns$.
 Therefore, our algorithm is faster when $s-r$ is small, and is slower 
 when $r<<s$.
   
 A simple case when our algorithm is faster is when the annihilator 
 is a monomial ideal. In this case $\ssl = B$, and our algorithm only reads 
 off $\border$, which costs at most $ns$, while this is the worst case for 
 all Mourrain's and Berthomieu-Faugere's algorithms. 
 Another set of sequences for which $s-r$ is small is when we {\it cut a 
 small corner from $\ssl$}. The following is an example of such a 
sequence.
\begin{example}
 Let $\ssl$ be the set of the monomials under the staircase of the 
 monomial ideal generated by $x^{4}, y^3$ and $xy^2$. Define the sequence 
 $\ell$  with the support $\ssl$ as follows.
 $\ell(x^4) =\ell(y)= d, \ell(x^4y)=\ell(y^2)=8$, and 
 for the rest of the monomials in $\ssl$, $\ell(x^iy^j)=i+j$. 
 $\ell$ has value zero for all monomials outside $\ssl$.
 Then the annihilator is generated by 
 \begin{equation}
   \border \cup \{
   x^4y^2, 
   x^3y^2-x^4+x^3y-x^2y^2, 
   x^4y-x^3y^2 \},
 \end{equation}
 and the quotient algebra of the annihilator ideal includes all 
 the monomials in $\ssl$ except for $x^4y^2$.
\end{example}

One can easily generalize the above example into sequences whose
support is defined by the staircase of certain ideals.

\begin{example}
  Consider the monomial ideal $\ideal{x^{d}, y^3,xy^2}$, where $d$ is
  a positive integer, and its staircase, i.e., the monomials that form
  a basis for the vector space of the quotient with respect to this
  ideal. Let $\ell(x^{d}) =\ell(y)= d, \ell(x^{d}y)=\ell(y^2)=2d$, and
  $\ell(x^iy^j)=i+j$ for the rest of the monomials in $\ssl$.  Then
  the quotient algebra of the annihilator ideal includes all of the
  monomials in $\ssl$ except for $x^dy^2$.
 \end{example}

 In Section \ref{section-benchmark}, we will present several examples of
 sequences with small $s-r$.

\begin{example}
 In the case that the annihilator is generated by 
 $\{x_i^{d_i}-x_j\ | \ 1\leq i \ne j\leq n \}$, 
 Mourrain's algorithm, as well as Berthomieu-Fauger's, 
 considers monomials 
 $1,x_i, \ldots,x_i^{d_i}$, while AnnihilatorViaDuality considers all 
 monomials whose partial degree at most 
 $d_i, \ i=1,\ldots,n$, except for $d_1+\cdots+d_m$ monomials in the basis 
 for the quotient vector space $\quot{\KK[X]}{\anni}$.
\end{example}

As mentioned earlier, the algorithms presented by Mora et al. in~\cite{mari-mora-moel-functionals} have complexity bound $\mathcal{O} (nt^3+nft^2)$, where $t$ is the number of given linear maps/sequences and $f$ is the cost of evaluation maps on monomials. These algorithms assume that the given sequences span the orthogonal of the annihilator as a subspace of the dual module. As there is a one-to-one correspondence between the generators of the dual vector space of an ideal and the vector space of the quotient of the polynomial ring with respect to that ideal, we have that $t \geq r$. So the complexity bound mentioned above becomes $\mathcal{O}(nr^3+nfr^2)$ if $t=r$, and moreover, if we assume that the cost of evaluation at each monomial is a constant, then the complexity bound becomes $\mathcal{O}(nr^3)$. 

One can check that based on our assumptions that the annihilator is $\mprimary$ and hence $s$ first term of the sequence is the non-zero terms, we have that $s\leq t$ and therefore, our algorithm is faster than Mora et al.'s algorithm. Our algorithm applied to $t$ sequences, where $t$ is as in Mora, et al.'s algorithms in~\cite{mari-mora-moel-functionals}, has complexity bound $\mathcal{O}((t-r)^3+nt)$ (because in this case having the first $t$ elements of the sequences is enough to determine the annihilator), which is faster than Mora et al.'s algorithms.
  
\subsection{Experiments}\label{section-benchmark}
  
 We have implemented computing annihilator using Hankel matrices and 
 AnnihilatorViaDuality algorithm in Maple. For latter, we have   used the 
 improved version of the Integration method for the dual computations
 \cite{angelos,mantzaflaris-rahkooy-zafeirakopoulos}, implemented in the 
 {\it mroot} package in Maple. Our implementation is available at 
 \url{https://github.com/rahkooy/dual-annihilator}. In the following, we 
 present a benchmark that we have obtained using AMD FX(tm)-6300 Six-Core, 
 with 16 GB RAM.
 
 For the benchmark, we have generated two sets  of sequences that are shown 
 in Tables \ref{tbl-randseqpol} and \ref{tbl-sequence-values}. The examples 
 are generated as follows:
    
 \begin{itemize}
  \item The sequences in Table \ref{tbl-randseqpol} are constructed as 
  follows. Let $\mathcal{M}$ be the set of monomials under the staircase of 
  a monomial ideal. We call the generators of this monomial ideal the 
  defining monomials of $\mathcal{M}$. Let 
  $(\ell_\alpha)_{\alpha \in \mathcal{M}}, \ \ell_\alpha \in \KK$ be a random 
  assignment of values in $\KK$ to the monomials in $\mathcal{M}$. For all 
  $\alpha \in \NN^n$, set $\ell_\alpha =0$. Then 
  $(\ell_\alpha)_{\alpha \in \NN^n} \in \KK^{\NN^n}$ and $\ssl = \mathcal{M}$, 
  where $\ssl$ is the finite set assigned to the the sequence $\ell$, which 
  is defined in Subsection \ref{prelim-dual}. Then $\ell$ satisfies the 
  conditions of the input of Algorithm \ref{algorithm}. 
 \item Sequences in Table \ref{tbl-sequence-values} are generated by 
  pre-assigned values to a given set of monomials $\mathcal{M}$, which is 
  closed under division.
 \end{itemize}
    
 Table \ref{tbl-hankel-dualann} includes the timings and memory size of the
 computations. Columns $H_\ell$ includes the size of the Hankel matrix,
 then the memory and CPU time used for computing the annihilator via that 
 Hankel matrix.  Column $I_\ell$ shows the size of the matrix in the 
 AnnihilatorViaDuality algorithm, and respective columns show the memory 
 used and the CPU time for this algorithm.
    
 One can see from Table \ref{tbl-hankel-dualann} that for  sequences    
 $J12$, $J3$, $J1$ and $l6$ AnnihilatorViaDuality algorithm has better 
 performance than computing the annihilator via Hankel matrices. For those 
 sequences, $s - r$ is equal to $250-214, 199-148, 79-54$ and $15-8$, 
 respectively. For sequences $J11$, $J9$, $J4$ and $\ell_5$ in 
 Table \ref{tbl-hankel-dualann}, AnnihilatorViaDuality algorithm much slower 
 performance than computing the annihilator using the Hankel matrices, as 
 $s-r$  for those sequences is $180-39, 80-48,24-9 ,8-1 $, respectively.

 \begin{table}[H]
   \centering \def\arraystretch{1.5} \setlength\tabcolsep{30pt}
      \caption{{\bf Random Sequences Over a Given Support $\mathcal{M}$}}
      \vspace{5mm}
      \label{tbl-randseqpol}
      \begin{tabular}{|c|l|}
      \hline
        \textbf{Sequence}& {\bf Monomials defining $\mathcal{M}$} \\
        \hline
        $J12$ & $x^{13},y^5,z^4,x^{12} z^2$ \\        
         \hline
        $J11$ & $x^{13},y^5,z^4,x^5 z^2$ \\
         \hline
        $J10$ & $x^8,y^5,z^4,x^5 z^2$ \\
         \hline
        $J9$ & $x^5,y^5,z^4,y^3z^2$ \\
         \hline
        $J6$ & $x^4,y^5,z^4,y^4z^3$\\
         \hline
        $J5$ & $x^2,y^5,z^4,y^3z^2$ \\   
         \hline
        $J4$ & $x^2,y^5,z^4,y z^2$\\ 
         \hline
        $J3$ & $x^{50},x^{49} y^3,y^4$\\ 
         \hline
        $J2$ & $x^{50},x^{20}y^3,y^4$\\   
         \hline
        $J1$ & $x^{20},x^{19}y^3,y^4$\\   
        \hline
       \end{tabular}
   \end{table}

   \begin{table}[H]
      \centering
      \def\arraystretch{1.5}
      \setlength\tabcolsep{15pt}
      \caption{{
       \bf Sequences Given by Their Values Over a Support $\mathcal{M}$}}
       \vspace{7mm}
      \label{tbl-sequence-values}
      \begin{tabular}{|c|l|}
      \hline
        \textbf{Sequence}& {\bf values over $\mathcal{M}$} \\
        \hline

         \hline
        $\ell_0 $ & $table([1=0,x=2,y=4,z=1,u=1,xy=1,xz=2,xu=1,$ \\
         & $yz=1,yu=1,  zu=1,xyz=1, xyu=1,xzu=1,yzu=1,xyzu=0$\\\\
         \hline
        $\ell_{11}$ & $table([1=0,x=2,y=4,z=1,u=1,xy=1,xz=0,xu=1,$ \\
        & $yz=1,yu=1,zu=1,xyz=1,xyu=1,xzu=1,yzu=1,xyzu=2]$\\
        \hline
        $\ell_{6}$ & $table([1=1,x=2,x^2=2,x^3=2,x^4=2,y=3,xy=3,x^2y=2,$\\
        & $x^3y=2,  x^4y=2,y^2=2,xy^2=2,x^2y^2=2,x^3y^2=2,x^4y^2=0]);$\\
        \hline
        $\ell_{5}$ & $table([1=0,x=2,y=4,z=1,xy=1,xz=0,yz=1,xyz=0])$\\   
        \hline
        $\ell_{1}$ & $table([1=2,x=2,y=1,xy=0, x^2=1,y^2=3])  $\\   
        \hline
      \end{tabular}
    \end{table}    
    
\begin{table}[H]
  \centering
  \def\arraystretch{1.5}
  \setlength\tabcolsep{4pt}
  \caption{
   {\bf 
   Computing Annihilator Using Hankel Matrices vs AnnihilatorViaDuality}
  }
  \vspace{8mm}
  \label{tbl-hankel-dualann}
  \begin{tabular}{|c|c|c|c|c|c|c|}
  \hline
   \textbf{sequence}& $\bm{H_\ell}$ & \textbf{memory} & {\bf CPU time} & 
   $\bm{I_\ell}$ & \textbf{memory} & {\bf CPU time} \\
  \hline
    $\ell_{11} $ & $16, 16$ & $1$MB & $20$ms & 
      $0$ (degenerate) & $255.58$KB & $4$ms \\
      \hline
   $J3$ & $199 \times 199$ & $2.26$GB & $21.65$sec &  $4 \times 4$ & 
    $5.84$MB & $76$ms \\
      \hline
    $J12$ & $250 \times 250$ & $4.70$GB & $32.89$sec & $103 \times 36$ & 
    $92.28$MB & $986$ms \\
      \hline
    $J1$ & $79 \times  79$ & $81.34$MB & $776$ms & $4 \times 4$ &
    $4.11$MB & $60$ms \\
      \hline
    $\ell_6$ & $15, 15$ & $0.72$MB & 
     $16$ms & $ 3, 3$ & $2.20$MB & $36$ ms \\    
      \hline
    $\ell_1$ & $6, 6$ & $1.74$MB & $31$ms & 
      $4, 4$ & $2.96$MB &$57$ms  \\
      \hline
    $I10$ & $11 \times 11$ & $0.54$MB & $12$ms  & $15 \times 8$ & 
    $4.36$MB & $68$ms \\
      \hline
    $\ell_0$ & $16,16$ & $1.03$MB & $20$ms & $38 \times 12$ & 
      $10.36$MB & $188$ms  \\
      \hline
    $\ell_5$ & $8,8$ & $417.05$KB & 
     $8$ms &  $17, 5$ & $3.55$MB & $56$ms \\
      \hline
    $J6$ & $76 \times 76$ & $77.02$MB & $593$ms & $86 \times29$ & 
    $76.57$MB & $929$ms \\    
      \hline
    $J5$ & $32 \times 32$ & $4.57$MB & $51$ms   & $82 \times25$ & 
    $41.26$MB & $494$ms  \\
      \hline
    $J10$ & $130 \times 130$ & $369.41$MB & $2.58$sec &  3$83 \times 119$ & 
    $1.35$GB & $16.72$sec\\    
      \hline
    $J2$  & $170 \times 170$ & $1.19$GB & $11.19$sec  & $1178 \times 50$ & 
    $5.17$GB & $108.52$sec \\
      \hline
    $J4$ & $24 \times 24$ & $2.18$MB & $35$ms & $121 \times 29$ & 
    $57.79$MB & $716$ms  \\
      \hline
    $J9$ & $80 \times 80$ & $78.21$MB & $570$ms & 
    $227 \times 73$ & $408.05$MB & $4.70$sec \\
      \hline
    $J11$ & $180 \times 180$ & $1.22$GB & $7.74$sec & $1435 \times 249$ & $19.22$GB & 
    $4.59$min \\
    \hline
   \end{tabular}
\end{table}
 
\begin{acknowledgment}
  The author would like to thanks E. Schost and A. Mantzaflaris for
  the discussions.
\end{acknowledgment}

\bibliographystyle{plain}
\bibliography{main.bib}

\end{document}